\documentclass[
  onecolumn,
  superscriptaddress,
  11pt,
  accepted=2017-06-09]{quantumarticle}

\pdfoutput=1

\usepackage[numbers,sort&compress]{natbib}

\usepackage[utf8]{inputenc}
\usepackage[english]{babel}
\usepackage[T1]{fontenc}
\usepackage{amsmath}
\usepackage{amssymb}
\usepackage{amsthm}
\usepackage{mathtools}

\usepackage{hyperref}
\usepackage{tikz}

\newtheorem{theorem}{Theorem}

\newcommand{\cA}{\mathcal{A}}
\newcommand{\cB}{\mathcal{B}}
\newcommand{\cE}{\mathcal{E}}
\newcommand{\cX}{\mathcal{X}}
\newcommand{\cY}{\mathcal{Y}}
\newcommand{\cL}{\mathcal{L}}
\newcommand{\cD}{\mathcal{D}}
\newcommand{\cG}{\mathcal{G}}
\newcommand{\cS}{\mathcal{S}}

\newcommand{\ip}[2]{\langle #1 , #2\rangle}

\newcommand{\ket}[1]{\ensuremath{\lvert #1 \rangle}} %
\newcommand{\bra}[1]{\ensuremath{\langle #1 \rvert}} %

\newcommand{\op}{\operatorname}
\newcommand{\tr}{\op{Tr}}

\usepackage{dsfont}
\newcommand{\I}{\mathds{1}}

\begin{document}

\title{On the complementary quantum capacity of the depolarizing channel}
\date{October 2, 2018}
\author{Debbie Leung}
\affiliation{Institute for Quantum Computing and Department of Combinatorics
  and Optimization, University of Waterloo}
\orcid{orcid.org/0000-0003-3750-2648}
\author{John Watrous}
\affiliation{Institute for Quantum Computing and School of Computer Science,
    University of Waterloo}
\orcid{orcid.org/0000-0002-4263-9393}

\maketitle

\begin{abstract}
  The qubit depolarizing channel with noise parameter $\eta$ transmits an input
  qubit perfectly with probability $1-\eta$, and outputs the completely mixed
  state with probability $\eta$.
  We show that its \emph{complementary channel} has \emph{positive} quantum
  capacity for all $\eta>0$.
  Thus, we find that there exists a single parameter family of channels having
  the peculiar property of having positive quantum capacity even when the
  outputs of these channels approach a fixed state independent of the input. 
  Comparisons with other related channels, and implications on the
  difficulty of studying the quantum capacity of the depolarizing
  channel are discussed.
\end{abstract}

\section{Introduction}

It is a fundamental problem in quantum information theory to determine the
capacity of quantum channels to transmit quantum information.
The \emph{quantum capacity} of a channel is the optimal rate at which one can
transmit quantum data with high fidelity through that channel when an
asymptotically large number of channel uses is made available.

In the classical setting, the capacity of a classical channel to transmit
classical data is given by Shannon's noisy coding theorem \cite{Shannon48}.
Although the error correcting codes that allow one to approach the
capacity of a channel may involve increasingly large block lengths, the capacity
expression itself is a simple, \emph{single letter formula} involving an 
optimization over input distributions maximizing the input/output mutual
information over \emph{one} use of the channel.

In the quantum setting, analyses inspired by the classical setting
have been performed \cite{Lloyd97,Shor02,Devetak05}, and an expression
for the quantum capacity has been found.  However, the capacity
expression involves an optimization similar to the classical setting
not for a single channel use, but for an increasingly large number of
channel uses.  The optimum value for $n$ copies of the channel leads
to the so-called \emph{$n$-shot coherent information} of the channel,
but little is known in general about how the $n$-shot coherent
information grows with $n$.  (Reference \cite{DiVincenzoSS98} showed
that the coherent information can be superadditive for some channels,
so the one-shot coherent information does not generally provide an
expression for the quantum capacity of a quantum channel.)
Consequently, the quantum capacity is unknown for many quantum
channels of interest.

Furthermore, \cite{DiVincenzoSS98} showed that the $n$-shot coherent
information of a channel can increase from zero to a positive quantity
as $n$ increases, and reference \cite{Cubitt2015unbounded} showed that
given any positive integer $n$, there is a channel whose $n$-shot
coherent information is zero but whose quantum capacity is
nevertheless positive.  Moreover, no algorithm is known to determine
if a quantum channel has zero or positive quantum capacity.  On the
other hand, some partial characterizations are known 
\cite{bennett1997-erasure,bruss1998-symmetric,Peres96,
2000-Horodecci-BE-channel,Smith-Smolin-2012-incapacity}.
For several well-known families of quantum channels that can be characterized by
noise parameters, the quantum capacity is proved to be zero within moderately
noisy regimes, well before the channel output becomes constant and independent
of the input.

In this paper, we show that any complementary channel to the qubit
depolarizing channel has positive quantum capacity (in fact, positive one-shot
coherent information) unless the output is exactly constant.
This is in sharp contrast with the superficially similar qubit depolarizing
channel and erasure channel, whose capacities vanish when the analogous noise
parameter is roughly half-way between the completely noiseless and noisy
extremes.
Prior to this work, it was not known (to our knowledge) that a family of quantum
channels could retain positive quantum capacity while approaching a channel
whose output is a fixed state, independent of the channel input.
We hope this example concerning how the quantum capacity does not vanish 
will shed light on a better characterization of when a channel has 
no quantum capacity. 

Another consequence of our result concerns the quantum capacity of low-noise
depolarizing channels.
Watanabe \cite{Watanabe2012} showed that if a given channel's complementary
channels have no quantum capacity, then the original channel must have quantum
capacity equal to its private classical capacity.
Furthermore, if the complementary channels have no classical private
capacity, then the quantum and private capacities are given by the one-shot 
coherent information.
Our result shows that Watanabe's results cannot be applied to the qubit
depolarizing channel. 
Very recently, \cite{LeditzkyLS17} established tight upper bounds on the
difference between the one-shot coherent information and the quantum and 
private capacities of a quantum channel, although whether or not the 
conclusion holds exactly remains open.  

In the remainder of the paper, we review background information concerning
quantum channels, quantum capacities, and relevant results on a few commonly
studied families of channels, and then prove our main results.

\section{Preliminaries}

Given a sender (Alice) and a receiver (Bob), one typically models
quantum communication from Alice to Bob as being sent through a
quantum channel $\Phi$.  We will associate the input and output
systems with finite-dimensional complex Hilbert spaces $\cA$ and
$\cB$, respectively.  In general, we write $\cL(\cX,\cY)$ to denote
the space of linear operators from $\cX$ to $\cY$, for
finite-dimensional complex Hilbert spaces $\cX$ and $\cY$, and we
write $\cL(\cX)$ to denote $\cL(\cX,\cX)$.  
For two operators $X,Y \in \cL(\cX)$, we use $\ip{X}{Y}$ to denote the
Hilbert-Schmidt inner product $\tr(X^{\ast}Y)$, where $X^{\ast}$ denotes the
adjoint of $X$.
We also write $\cD(\cX)$ to denote the set of positive semidefinite,
trace one operators (i.e., density operators) acting on $\cX$.

A quantum channel $\Phi$ from Alice to Bob is a completely positive,
trace-preserving linear map of the form
\begin{equation}
  \Phi:\cL(\cA)\rightarrow\cL(\cB) \,.
\end{equation}
There exist several well-known characterizations of quantum channels.  
The first one we need is given by the Stinespring representation, in which
a channel $\Phi$ is described as
\begin{equation}
  \label{eq:Stinespring-Phi}
  \Phi(\rho) = \tr_{\cE} (A \rho A^{\ast}),
\end{equation}
where $\cE$ is a finite-dimensional complex Hilbert space representing
an ``environment'' system, $A\in\cL(\cA,\cB\otimes\cE)$ is an isometry (i.e., a
linear operator satisfying $A^{\ast} A = \I$), and
$\tr_{\cE}:\cL(\cB\otimes\cE)\rightarrow\cL(\cB)$ denotes the partial trace
over the space $\cE$.
In this context, the isometry $A$ is sometimes known as an isometric extension
of $\Phi$, and is uniquely determined up to left multiplication by an isometry
acting on $\cE$.

For a channel $\Phi$ with a Stinespring representation 
\eqref{eq:Stinespring-Phi}, the channel $\Psi$ of the form
$\Psi:\cL(\cA)\rightarrow\cL(\cE)$ that is given by
\begin{equation}
  \Psi(\rho) = \tr_{\cB} (A \rho A^{\ast})
\end{equation}
is called a \emph{complementary channel} to $\Phi$.  
Following the degree of freedom in the Stinespring representation, a
complementary channel of $\Phi$ is uniquely determined up to an isometry on the
final output.
A channel $\Psi$ that is complementary to $\Phi$ may be viewed as representing
information that leaks to the environment when $\Phi$ is performed.

The second type of representation we need is a Kraus representation
\begin{equation}
  \Phi(\rho) = \sum_{k = 1}^N A_k \rho A_k^{\ast} \,,
\end{equation}
where the operators $A_1,\ldots,A_N \in\cL(\cA,\cB)$ (called Kraus operators)
satisfy
\begin{equation}
  \sum_{k = 1}^N A_k^{\ast} A_k = \I \,.
\end{equation}

The \emph{coherent information} of a state $\rho\in\cD(\cA)$ through a channel
$\Phi:\cL(\cA)\rightarrow\cL(\cB)$ is defined as
\begin{equation}
  \op{I}_{\textup{\tiny C}}(\rho;\Phi)
  = \op{H}(\Phi(\rho)) - \op{H}(\Psi(\rho)) \,,
\end{equation}
for any channel $\Psi$ complementary to $\Phi$, where 
$\op{H}(\sigma) = - \tr(\sigma \log \sigma)$ denotes the von~Neumann entropy of
a density operator $\sigma$.
Note that the coherent information is independent of the choice of the
complementary channel $\Psi$.
The coherent information of $\Phi$ is given by the maximum over all inputs 
\begin{equation}
\op{I}_{\textup{\tiny C}}(\Phi) 
= \max_{\rho \in \cD(\cA)} \op{I}_{\textup{\tiny C}}(\rho;\Phi) \,.
\end{equation}
The $n$-shot coherent information of $\Phi$ is 
$\op{I}_{\textup{\tiny C}}(\Phi^{\otimes n})$.  
The \emph{quantum capacity theorem} \cite{Lloyd97,Shor02,Devetak05}
states that the quantum capacity of $\Phi$ is given by the expression
\begin{equation}
  \op{Q}(\Phi) = \lim_{n\rightarrow\infty} 
  \frac{\op{I}_{\textup{\tiny C}}(\Phi^{\otimes n})}{n} \,.
\label{eq:qcap}
\end{equation}
The $n$-shot coherent information
$\op{I}_{\textup{\tiny C}}(\Phi^{\otimes n})$ of a channel $\Phi$ is trivially
lower-bounded by $n$ times the coherent information
$\op{I}_{\textup{\tiny C}}(\Phi)$, and therefore the coherent information
$\op{I}_{\textup{\tiny C}}(\Phi)$ provides a lower-bound on the quantum
capacity of $\Phi$.

The \emph{qubit depolarizing channel} with noise parameter $\eta$, 
denoted by $\Phi_\eta$, takes a qubit state $\rho \in \cD(\mathbb{C}^2)$
to itself with probability $1-\eta$, and replaces it with a random output with
probability $\eta$:
\begin{equation}
  \Phi_{\eta}(\rho) = (1 - \eta) \, \rho + \eta \, \frac{\I}{2} \,. 
\end{equation}
One Kraus representation of $\Phi_{\eta}$ is
\begin{equation}
  \Phi_{\eta}(\rho) = (1 - \varepsilon) \, \rho + \frac{\varepsilon}{3}
  \bigl(\sigma_1 \, \rho \, \sigma_1 + \sigma_2 \, \rho \, \sigma_2 
  + \sigma_3 \, \rho \, \sigma_3 \bigr),
\end{equation}
where $\varepsilon = 3\eta/4$, and 
\begin{equation}
\sigma_1 = \begin{pmatrix}
    0 & 1\\ 1 & 0
  \end{pmatrix},\quad
  \sigma_2 = \begin{pmatrix}
    0 & -i\\ i & 0
  \end{pmatrix},
  \quad\text{and}\quad
  \sigma_3 = \begin{pmatrix}
    1 & 0\\ 0 & -1
  \end{pmatrix}
\end{equation}
denote the Pauli operators.
A Stinespring representation of $\Phi_{\eta}$ that corresponds naturally to this
Kraus representation is
\begin{equation}
  \Phi_\eta(\rho) = \tr_{\cE} 
  \bigl(A_{\varepsilon} \rho A_{\varepsilon}^{\ast}\bigr)
\end{equation}
for the isometric extension 
\begin{equation}
  A_{\varepsilon} = \sqrt{1 - \varepsilon}\, \I \otimes \ket{0}
  + \sqrt{\frac{\varepsilon}{3}} \bigl( \sigma_1 \otimes \ket{1}
  + \sigma_2 \otimes \ket{2}
  + \sigma_3 \otimes \ket{3}\bigr).
\end{equation}
The complementary channel $\Psi_{\eta}$ to $\Phi_{\eta}$ determined by this
Stinespring representation is given by
\begin{equation}
  \Psi_{\eta}(\rho) = 
  \begin{pmatrix}
    1- \varepsilon & 
    \sqrt{\frac{\varepsilon(1-\varepsilon)}{3}} \ip{\sigma_1}{\rho} & 
    \sqrt{\frac{\varepsilon(1-\varepsilon)}{3}} \ip{\sigma_2}{\rho} & 
    \sqrt{\frac{\varepsilon(1-\varepsilon)}{3}} \ip{\sigma_3}{\rho} \\[2mm]
    \sqrt{\frac{\varepsilon(1-\varepsilon)}{3}} \ip{\sigma_1}{\rho} & 
    \frac{\varepsilon}{3} &
    - \frac{i \varepsilon}{3} \ip{\sigma_3}{\rho} &
    \frac{i \varepsilon}{3} \ip{\sigma_2}{\rho} \\[2mm]
    \sqrt{\frac{\varepsilon(1-\varepsilon)}{3}} \ip{\sigma_2}{\rho} & 
    \frac{i \varepsilon}{3} \ip{\sigma_3}{\rho} &
    \frac{\varepsilon}{3} &
    - \frac{i \varepsilon}{3} \ip{\sigma_1}{\rho} \\[2mm]
    \sqrt{\frac{\varepsilon(1-\varepsilon)}{3}} \ip{\sigma_3}{\rho} & 
    - \frac{i \varepsilon}{3} \ip{\sigma_2}{\rho} &
    \frac{i \varepsilon}{3} \ip{\sigma_1}{\rho} &
    \frac{\varepsilon}{3}
  \end{pmatrix}\!.
\label{eq:epolarizing}
\end{equation}
We call this complementary channel the \emph{epolarizing channel}.
Note that when $\eta \approx 0$, the channel $\Phi_\eta$ is nearly noiseless,
while $\Psi_\eta$ is very noisy, and the opposite holds when $\eta \approx 1$.

We will use the expressions above to calculate a lower-bound on the
coherent information $\op{I}_{\textup{\tiny C}}(\Psi_{\eta})$, which
provides a lower-bound on the quantum capacity of the epolarizing
channel $\Psi_\eta$.  

\section{Main result} 

\begin{theorem}
  \label{theorem:positive-complementary-coherent-information}
  Let $\Phi_{\eta}$ be the qubit depolarizing channel with noise
  parameter $\eta \in [0,1]$. Any complementary channel to
  $\Phi_{\eta}$ has positive coherent information when $\eta > 0$.
\end{theorem}

\begin{proof}
  The coherent information is independent of the choice of the complementary
  channel, so it suffices to focus on the choice $\Psi_{\eta}$ described in
  \eqref{eq:epolarizing}.
  Taking
  \begin{equation}
   \label{eq:input}
    \rho = \begin{pmatrix}
      1 - \delta & 0\\
      0 & \delta
    \end{pmatrix}
  \end{equation}
  yields $\ip{\sigma_1}{\rho} = 0$, $\ip{\sigma_2}{\rho} = 0$, and
  $\ip{\sigma_3}{\rho} = 1 - 2\delta$, and therefore
  \begin{equation}
    \Psi_{\eta}(\rho) = 
    \begin{pmatrix}
      (1- \varepsilon) & 0 & 0 & 
      \sqrt{\frac{\varepsilon(1-\varepsilon)}{3}} (1 - 2\delta) \\[2mm]
      0 & \frac{\varepsilon}{3} & -\frac{i \varepsilon}{3} (1 - 2\delta) 
      & 0 \\[2mm]
      0 & \frac{i \varepsilon}{3} (1 - 2\delta) & \frac{\varepsilon}{3} & 
      0 \\[2mm]
      \sqrt{\frac{\varepsilon(1-\varepsilon)}{3}} (1 - 2\delta) & 0 & 0 & 
      \frac{\varepsilon}{3}
    \end{pmatrix}.
  \end{equation}
  A closed-form expression for the entropy of $\Psi_{\eta}(\rho)$ is not 
  difficult to obtain; however for our purpose it suffices to lower bound 
  $H(\Psi_{\eta}(\rho))$ with the following simple argument.  
  Define the state 
  \begin{equation}
    \xi = 
    \begin{pmatrix}
    (1-\varepsilon) & 0 & 0 & \sqrt{\frac{\varepsilon(1-\varepsilon)}{3}}\\[2mm]
     0 & \frac{\varepsilon}{3} & -\frac{i\varepsilon}{3} (1-2\delta) & 0 \\[2mm]
     0 & \frac{i\varepsilon}{3} (1-2\delta) & \frac{\varepsilon}{3} & 0 \\[2mm]
    \sqrt{\frac{\varepsilon(1-\varepsilon)}{3}} & 0 & 0 & \frac{\varepsilon}{3}
    \end{pmatrix}\,, 
  \end{equation}
  and note that
  \begin{equation}
    \Psi_{\eta}(\rho) = (1-\delta) \, \xi + \delta \, U \xi U^{\ast}
  \end{equation}
  where $U$ is diagonal with diagonal entries $(1,1,1,-1)$. 
  As the von~Neumann entropy is concave and invariant under unitary
  conjugations, it follows that $H(\Psi_{\eta}(\rho)) \geq H(\xi)$.  
  Finally, $\xi$ has eigenvalues
  \begin{equation}
    \biggl\{ 1-\frac{2\varepsilon}{3},0,\frac{2\varepsilon(1-\delta)}{3},
    \frac{2\varepsilon\delta}{3}\biggr\}
    = \biggl\{ 1-\frac{\eta}{2},0,\frac{\eta(1-\delta)}{2},
    \frac{\eta\delta}{2}\biggr\}
  \end{equation}
  and entropy   
  \begin{equation}
    \op{H}(\xi)  
    = \frac{\eta}{2} \op{H}_2(\delta) + \op{H}_2\Bigl(\frac{\eta}{2}\Bigr).
    \label{eq:eveentropy}
  \end{equation}
  
  On the other hand, 
  \begin{equation}
    \Phi_{\eta}(\rho) = \begin{pmatrix}
      (1 - \eta)(1 - \delta) + \frac{\eta}{2} & 0\\
      0 & (1 - \eta) \, \delta + \frac{\eta}{2} 
    \end{pmatrix},
  \end{equation}
  and therefore
  \begin{equation}
    \op{H}\bigl(\Phi_{\eta}(\rho)\bigr) 
    = \op{H}_2\Bigl((1 - \eta)\,\delta + \frac{\eta}{2}\Bigr).
 \label{eq:bobentropy}
 \end{equation}
  By the mean value theorem, one has
  \begin{equation}
    \op{H}_2\Bigl((1 - \eta)\,\delta + \frac{\eta}{2}\Bigr)
    - \op{H}_2\Bigl(\frac{\eta}{2}\Bigr)
    = (1 - \eta)\,\delta \, \bigl(\log(1-\mu) -\log(\mu) \bigr)
  \end{equation}
  for some choice of $\mu$ satisfying
  $\eta/2 \leq \mu \leq (1 - \eta)\,\delta + \eta/2$, and therefore
  \begin{equation}
    \op{H}\bigl(\Phi_{\eta}(\rho)\bigr)
    \leq \op{H}_2\Bigl(\frac{\eta}{2}\Bigr)
    + (1 - \eta)\,\delta \log\Bigl(\frac{2}{\eta}\Bigr).
  \end{equation}
  Therefore, the coherent information of $\rho$ through
  $\Psi_{\eta}$ is lower-bounded as follows:
  \begin{equation}
    \begin{multlined}
      \op{I}_{\textup{\tiny C}}(\rho;\Psi_\eta)
      = \op{H}(\Psi_\eta(\rho)) - \op{H}(\Phi_\eta(\rho))\\
      \geq \op{H}(\xi) - \op{H}(\Phi_\eta(\rho))
      \geq \frac{\eta}{2}\op{H}_2(\delta) - 
      (1{-}\eta)\,\delta\log\Bigl(\frac{2}{\eta}\Bigr).
    \end{multlined}
  \end{equation}
  We solve the inequality where the rightmost expression is strictly positive.
  The values of $\delta$ for which strict positivity holds includes the
  interval
  \begin{equation}
    0 < \delta \leq 2^{-\frac{2(1 - \eta)}{\eta}\log\bigl(\frac{2}{\eta}\bigr)},
  \end{equation}
  which completes the proof.  
\end{proof}

Note that one can obtain a closed-form expression of
$\op{I}_{\textup{\tiny C}}(\rho;\Psi_\eta)$ for $\rho$ given by
\eqref{eq:input}.  Furthermore, this input is optimal due to
the symmetry of $\Psi_\eta$.  Therefore, the actual coherent
information of $\Psi_\eta$ can be obtained by optimizing
$\op{I}_{\textup{\tiny C}}(\rho;\Psi_\eta)$ over $\delta$.  This
method does not extend to the calculation of the $n$-shot coherent
information, nor the asymptotic quantum capacity of $\Psi_\eta$.

\section{Comparisons with some well-known families of channels} 

The \emph{qubit erasure channel} with noise parameter $\eta \in [0,1]$,
denoted by $\Xi_\eta$, takes a single qubit state $\rho \in \cD(\mathbb{C}^2)$
to itself with probability $1-\eta$, and replaces it by an error symbol
orthogonal to the input space with probability $\eta$.
The quantum capacity of the erasure channel is known and is given by
$\op{Q}(\Xi_\eta) = \max(0,1-2\eta)$ \cite{bennett1997-erasure}.

We can relate the depolarizing channel, the erasure channel, and the
epolarizing channel as follows.
Let each of $\cA, \cS_1, \cS_2, \cG_1, \cG_2$ denote a qubit system.
Consider an isometry 
\begin{equation}
  A \in \cL(\cA, \cS_1 \otimes \cS_2 \otimes \cG_1
  \otimes \cG_2 \otimes \cA)
\end{equation}
acting on a pure qubit state $\ket{\psi} \in \cA$ as
\begin{equation}
  \ket{\psi}_{\cA} \mapsto \left[
    \ket{0}\bra{0}_{\cS_1} \otimes \I_{\cA \cG_1} +  
    \ket{1}\bra{1}_{\cS_1} \otimes \textsc{swap}_{\cA \cG_1} 
    \right]
  \ket{s}_{\cS_1 \cS_2} \ket{\Phi}_{\cG_1 \cG_2} \ket{\psi}_{\cA},
\end{equation}
where $\ket{s} = \sqrt{1-\eta} \, \ket{00} + \sqrt{\eta} \, \ket{11}$ and
$\ket{\Phi} = \frac{1}{\sqrt{2}} ( \ket{00} + \ket{11})$, and where
the subscripts denote the pertinent systems.
The isometry can be interpreted as follows.
System $\cA$ (the input space) initially contains the input state
$\ket{\psi}_{\cA}$, while a system $\cG_1$ (which represents a ``garbage''
space) is initialized to a completely mixed state.
The input is swapped with the garbage if and only if a measurement of the
$\cS_1$ system (which represents a ``syndrome'') causes the state $\ket{s}$
of $\cS_1\cS_2$ to collapse to $\ket{11}$.
Finally, each of the depolarizing, erasure, and the epolarizing channel can be 
generated by discarding a subset of the systems as follows:
\begin{equation}
  \label{eq:all3}
  \begin{aligned}
    \Phi_\eta(\rho) & = \tr_{\cS_1 \otimes \cS_2 \otimes \cG_1
      \otimes \cG_2} (A \rho A^{\ast})\,,\\
    \Xi'_\eta(\rho) & = \tr_{\cS_2 \otimes \cG_1 \otimes \cG_2}
    (A \rho A^{\ast})\,,\\
    \Psi'_\eta(\rho) & = \tr_{\cA} (A \rho A^{\ast}) \,.
  \end{aligned}
\end{equation} 
To be more precise, the channel $\Xi'_\eta$ in \eqref{eq:all3} is related
to the channel $\Xi_\eta$ described earlier by an isometry---for all
relevant purposes, $\Xi'_\eta$ and $\Xi_\eta$ are equivalent.
Likewise, $\Psi'_\eta$ is equivalent to $\Psi_\eta$ in \eqref{eq:epolarizing}.
If we ignore the precise value of $\eta$, the systems $\cA$ and
$\cG_1$ carry qualitatively similar information.  Furthermore, the
additional garbage system $\cG_2$ is irrelevant. 
So, the three families of channels are distinguished by which syndrome
systems are available in the output: none for the depolarizing channel output,
both for the epolarizing channel, and one for the erasure channel.
These different possibilities cause significant differences in the noise
parameter ranges for which the quantum capacity vanishes
\cite{bennett1997-erasure,DiVincenzoSS98}:
\begin{equation}
  \begin{aligned}
    \op{Q}(\Phi_\eta) = 0 & \quad \text{if} \;\; 1/3 \leq \eta \leq 1 \,,\\
    \op{Q}(\Xi_\eta) = 0 & \quad \text{iff} \;\; 1/2 \leq \eta \leq 1 \,,\\ 
    \op{Q}(\Psi_\eta) = 0 & \quad \rm{iff} \;\; \eta = 0 \,.
  \end{aligned}
\end{equation}
In particular, when $\eta \approx 0$, the syndrome state carries very
little information and only interacts weakly with the input---and yet
having all shares of it in the output keeps the quantum capacity of
the epolarizing channel positive.
The syndrome systems therefore carry qualitatively significant information 
that is quantitatively negligible.
Despite recent results in \cite{LeditzkyLS17}, the extent to which 
this phenomenon is relevant to an understanding of the
capacity of the depolarizing channel is a topic for further research.

We also note that the qubit amplitude damping channel (see
\cite{NC00}) has vanishing quantum capacity if and only if the noise
parameter satisfies $1/2 \leq \eta \leq 1$, which is similar to the erasure
channel (while the output only approaches a constant as $\eta
\rightarrow 1$).  
The dephasing channel (see below) does not take the
input to a constant for all noise parameters.  

\section{Extension to other channels}

A mixed Pauli channel on one qubit can be described by a Kraus representation
\begin{equation}
  \Theta(\rho) 
  = (1-p_1-p_2-p_3) \; \rho + p_1 \, \sigma_1 \, \rho \, \sigma_1
  + p_2 \, \sigma_2 \, \rho \, \sigma_2
  + p_3 \, \sigma_3 \, \rho \, \sigma_3 \,,
\end{equation}
for $p_1,p_2,p_3\geq 0$ satisfying $p_1 + p_2 + p_3\leq 1$.
For example, a \emph{dephasing channel} can be described in this way
by taking $p_1 = p_2 = 0$ and $p_3 \in [0,1]$.
In this case the quantum capacity is known to equal $1-H_2(p_3)$,
which is positive except when $p_3 = 1/2$.
Any complementary channel of such a dephasing channel must have zero quantum
capacity.
If at least $3$ of the $4$ probabilities $(1-p_1-p_2-p_3),p_1, p_2, p_3$ are 
positive, a generalization of our main result demonstrates that the capacity of
a complementary channel of $\Theta$ has positive coherent information, as is
proved below, so the phenomenon exhibited by the depolarizing channel is
therefore not an isolated instance.
It is an interesting open problem to determine which mixed unitary channels
in higher dimensions, meaning those channels having a Kraus representation in
which every Kraus operator is a positive scalar multiple of a unitary operator,
have complementary channels with positive capacity.
(It follows from the work of \cite{CubittRS08} that every mixed unitary channel
with commuting Kraus operators is degradable, and therefore must have zero
complementary capacity.)

\begin{theorem}
  \label{theorem:mixed-pauli}
  Consider the mixed Pauli channel on one qubit described by
  \begin{equation}
    \Theta(\rho) = p_0 \; \rho + p_1 \, \sigma_1 \, \rho \, \sigma_1
    + p_2 \, \sigma_2 \, \rho \, \sigma_2
    + p_3 \, \sigma_3 \, \rho \, \sigma_3 \,,
  \end{equation}
  where $p_0,p_1,p_2,p_3\geq 0$, $p_0+p_1+p_2+p_3=1$.
  If three or more of these probabilities are nonzero, then any complementary
  channel to $\Theta$ has positive coherent information.  
\end{theorem}

\begin{proof}[Proof of Theorem~\ref{theorem:mixed-pauli}]
The proof is similar to that of 
Theorem~\ref{theorem:positive-complementary-coherent-information}. 
We can assume without loss of generality that $p_0 \geq p_1 \geq p_2
\geq p_3$, by redefining the basis of the output space if necessary.  
A convenient choice of the isometric extension is 
\begin{equation}
  A = \sum_{i=0}^3 \sqrt{p_i} \, \sigma_i \otimes \ket{i} \,,
\end{equation}
where $\sigma_0 = \I$.
This gives a complementary channel $\Theta^c$ acting as 
\begin{equation}
  \Theta^c(\rho) = 
  \begin{pmatrix}
    p_0\!\! & 
    \sqrt{p_0 p_1} \, \ip{\sigma_1}{\rho}\!\! & 
    \sqrt{p_0 p_2} \, \ip{\sigma_2}{\rho}\!\! & 
    \sqrt{p_0 p_3} \, \ip{\sigma_3}{\rho}\!\! \\[2mm]
    \sqrt{p_0 p_1} \, \ip{\sigma_1}{\rho} & 
    p_1\!\! &
    -{i} \sqrt{p_1 p_2} \, \ip{\sigma_3}{\rho}\!\! &
      i \sqrt{p_1 p_3} \, \ip{\sigma_2}{\rho} \\[2mm]
    \sqrt{p_0 p_2} \, \ip{\sigma_2}{\rho}\!\! & 
    i \sqrt{p_1 p_2} \, \ip{\sigma_3}{\rho}\!\! &
    p_2\!\! &
    -{i} \sqrt{p_2 p_3} \, \ip{\sigma_1}{\rho}\! \\[2mm]
    \sqrt{p_0 p_3} \, \ip{\sigma_3}{\rho}\!\! & 
    -{i} \sqrt{p_1 p_3} \, \ip{\sigma_2}{\rho}\!\! &
    i \sqrt{p_2 p_3} \, \ip{\sigma_1}{\rho}\!\! &
    p_3\!\!
  \end{pmatrix}.
\end{equation}

We choose the following parametrization to simplify the analysis.
Let $p_1 = p > 0$, $p_2 = \alpha p$ where $0 < \alpha \leq 1$, and 
$\eta' = 2(1+\alpha)p$.
We will see that the parameter $\eta'$ enters the current proof in a way
that is similar to the noise parameter $\eta$ for the depolarizing channel
in the proof of 
Theorem~\ref{theorem:positive-complementary-coherent-information}.
Once again, we take 
\begin{equation}
  \rho = \begin{pmatrix}
    1 - \delta & 0\\
    0 & \delta
  \end{pmatrix}
\end{equation}
so $\ip{\sigma_1}{\rho} = 0$, $\ip{\sigma_2}{\rho} = 0$, and
$\ip{\sigma_3}{\rho} = 1 - 2\delta$, and therefore
\begin{equation}
  \Theta^c(\rho) = 
  \begin{pmatrix}
    p_0 & 0 & 0 & 
    \sqrt{p_0 p_3} (1 - 2\delta) \\[2mm]
    0 & p_1 & -{i} \sqrt{p_1 p_2} (1 - 2\delta) & 0 \\[2mm]
    0 & i \sqrt{p_1 p_2} (1 - 2\delta) & p_2 & 0 \\[2mm]
    \sqrt{p_0 p_3} (1 - 2\delta) & 0 & 0 & p_3
  \end{pmatrix}.
\end{equation}
The entropy of $\Theta^c(\rho)$ is at least the entropy of the state 
\begin{equation}
  \xi' = 
  \begin{pmatrix}
    p_0 & 0 & 0 & 
    \sqrt{p_0 p_3} \\[2mm]
    0 & p_1 & -i\sqrt{p_1 p_2} (1 - 2\delta) & 0 \\[2mm]
    0 & i\sqrt{p_1 p_2} (1 - 2\delta) & p_2 & 0 \\[2mm]
    \sqrt{p_0 p_3} & 0 & 0 & p_3
  \end{pmatrix}.
\end{equation}
The submatrix at the four corners gives rise to the eigenvalues
$\{p_0 + p_3,0\} = \{1-\frac{\eta'}{2},0\}$ as in the proof of
Theorem~\ref{theorem:positive-complementary-coherent-information}.
Meanwhile, the middle block can be rewritten as 
\begin{equation}
  \frac{\eta'}{2} \begin{pmatrix}
    \frac{1}{1+\alpha} & \frac{i\sqrt{\alpha}}{1+\alpha}(1-2\delta) \\[2mm]
    \frac{-i\sqrt{\alpha}}{1+\alpha}(1-2\delta) & \frac{\alpha}{1+\alpha}
  \end{pmatrix} 
  = 
  \frac{\eta'}{2} \begin{pmatrix}
    \frac{1}{2} + \frac{\cos (2\theta)}{2} & 
    \frac{i \sin (2 \theta)}{2} (1-2\delta) \\[2mm]
    \frac{-i \sin (2 \theta)}{2} (1-2\delta) & 
    \frac{1}{2} - \frac{\cos (2\theta)}{2} 
  \end{pmatrix} \, ,
  \label{eq:midblock}
\end{equation}
where 
\begin{align}
  \frac{1}{1+\alpha} = \cos^2(\theta)
  = \frac{1}{2} + \frac{\cos(2\theta)}{2}\,,\\
  \frac{\alpha}{1+\alpha} = \sin^2 (\theta) = \frac{1}{2} - 
  \frac{\cos (2\theta)}{2} \,,\\
  \frac{\sqrt{\alpha}}{1+\alpha} = \sin(\theta)\cos(\theta)
  = \frac{\sin (2 \theta)}{2} \,,
\end{align}
and $0< \theta \leq \frac{\pi}{2}$.
From equation \eqref{eq:midblock}, the eigenvalues of the middle block can be
evaluated as
\begin{equation}
  \frac{\eta'}{2}\biggl\{\frac{1+r}{2},\frac{1-r}{2}\biggr\}
\end{equation}
where 
\begin{equation} r^2 = \cos^2 (2 \theta) + (1-2 \delta)^2 \sin^2(2 \theta)
  = 1-4\delta \sin^2 (2 \theta) + 4 \delta^2 \sin^2 (2 \theta) \,.
\end{equation}
If we define the variable $\delta'$ to satisfy the equation
\begin{equation}
  \delta (1-\delta) \sin^2 (2 \theta) = \delta' (1-\delta'),
\end{equation}
then $r = 1-2\delta'$ and the two eigenvalues are
\begin{equation}
  \biggl\{\frac{\eta'(1-\delta')}{2},\frac{\eta'\delta'}{2}\biggr\}.
\end{equation}

Altogether, the spectrum of $\xi'$ is
\begin{equation}
\biggl\{ 1-\frac{\eta'}{2},0,\frac{\eta'(1-\delta')}{2},
\frac{\eta'\delta'}{2}\biggr\},
\end{equation}
which has the same form as the spectrum of $\xi$ in the proof of 
Theorem~\ref{theorem:positive-complementary-coherent-information}, 
and the entropy of $\xi'$ is analogous to \eqref{eq:eveentropy}, 
\begin{equation}
  \op{H}(\xi')  
  = \frac{\eta'}{2} \op{H}_2(\delta') + \op{H}_2\biggl(\frac{\eta'}{2}\biggr).
\end{equation}
On the other hand, $\Theta(\rho)$ has exactly the same expression as
$\Phi_{\eta'}(\rho)$ and the entropy of $\Theta(\rho)$ is analogous to
\eqref{eq:bobentropy},
\begin{equation}
  \op{H}\bigl(\Theta(\rho)) 
  = \op{H}_2\biggl((1 - \eta')\,\delta + \frac{\eta'}{2}\biggr).
\end{equation}
Following arguments similar to the proof of
Theorem~\ref{theorem:positive-complementary-coherent-information},
the coherent information of $\rho$ through
  $\Theta^c$ is lower-bounded as follows:
  \begin{equation}
    \op{I}_{\textup{\tiny C}}(\rho;\Theta^c)
    = \op{H}(\Theta^c(\rho)) - \op{H}(\Theta(\rho))
    \; \geq \;
    \op{H}(\xi') - \op{H}(\Theta(\rho))
    \; \geq \;  
    \frac{\eta'}{2}\op{H}_2(\delta') - 
    (1{-}\eta')\,\delta\log\Bigl(\frac{2}{\eta'}\Bigr).
  \end{equation}
We have a $\delta'$-dependency in the first term and
$\delta$-dependency in the second term.
However,
\begin{equation}
  \delta (1-\delta) \sin^2 (2 \theta) = \delta' (1-\delta'),
\end{equation}
and $\sin^2 (2 \theta)$ is a positive constant determined by $\alpha = p_2/p_1$,
so for sufficiently small $\delta$, the above equation is strictly positive.
\end{proof}

\section{Conclusion}

We have shown that any complementary channel to the qubit depolarizing channel
has positive quantum capacity unless its output is exactly constant.  
This gives an example of a family of channels whose outputs approach a
constant, yet retain positive quantum capacity.
We also point out a crucial difference between the epolarizing channel and the
related depolarizing and erasure channels.
We hope these observations will shed light on what may or may not cause the
quantum capacity of a channel to vanish.

Our work also rules out the possibility that Watanabe's results
\cite{Watanabe2012} can be applied directly to show that the low-noise
depolarizing channel has quantum capacity given by the $1$-shot coherent
information.
Very recently, \cite{LeditzkyLS17} established tight upper bounds on the
difference between the one-shot coherent information and the quantum and 
private capacities of a quantum channel.
While our results do not have direct implications to these capacities 
of $\Phi_{\eta}$, we hope they provide insights for further investigations 
beyond the bounds established in \cite{LeditzkyLS17}.  

\subsection*{Acknowledgements}

We thank Ke Li, Graeme Smith, and John Smolin for inspiring discussions on the
depolarizing channel, and we thank the hospitality of the Physics of
Information Group at IBM TJ Watson Research Center.
We also thank Frederic Dupuis, Aram Harrow, William Matthews, Graeme Smith,
Mark Wilde, and Andreas Winter for a lively discussion concerning the
epolarizing channel during the workshop \emph{Beyond IID in Information
Theory}, 5--10 July 2015, and the hospitality of The Banff International
Research Station (BIRS).
Finally, we thank Yuan Su for bringing an error in a previous version of this
paper to our attention.
This research was supported by NSERC, the Canadian Institute for
Advanced Research, and the Canada Research Chairs Program.

\bibliographystyle{plainnat}

\end{document}